\documentclass[11pt]{article}

\usepackage{sn-preamble} 
\usepackage{tikz}
\usepackage{verbatim}
\usepackage{cancel}

\newcommand{\ignore}[1]{}

\newcommand{\OVB}{\mathsf{Online Vector Balancing}}

\title{A Simple Algorithm for Dynamic Carpooling with Recourse}
\author{
    Yuval Efron \\
    \small Columbia University\\
    \small \texttt{efronyuv@gmail.com}\\
    \and
    Shyamal Patel\thanks{Research supported in part  by NSF grants  CCF-2106429, CCF-2107187, CCF-2218677, ONR grant ONR-13533312, and a NSF Graduate Student Fellowship.}\\
    \small Columbia University\\
    \small \texttt{shyamalpatelb@gmail.com}\\
    \and
    Cliff Stein\thanks{Research supported in part  by NSF grant CCF-2218677, ONR grant ONR-13533312, and by the Wai T. Chang Chair in Industrial Engineering and Operations Research.} \\
    \small Columbia University\\
    \small \texttt{cliff@ieor.columbia.edu}\\
}
\date{}
\begin{document}

\pagenumbering{gobble}
\maketitle  

\pagenumbering{arabic}

\begin{abstract}
    We give an algorithm for the fully-dynamic carpooling problem with recourse: Edges arrive and depart online from a graph $G$ with $n$ nodes according to an adaptive adversary. Our goal is to maintain an orientation $H$ of $G$ that keeps the discrepancy, defined as $\max_{v \in V} |\deg_H^+(v) - \deg_H^-(v)|$, small at all times.
    
    We present a simple algorithm and analysis for this problem with recourse based on cycles that simplifies and improves on a result of Gupta et al. [SODA '22].
\end{abstract}

\section{Introduction}
\paragraph{}{In this paper, we consider the graphical carpooling 
 problem: Given a graph $G = (V,E)$, find an orientation $H$ of $G$ that minimizes the \emph{discrepancy}, which is the absolute difference between the in-degree and out-degree of any vertex i.e. $\max_{v \in V} |\deg^-_H(v) - \deg^+_H(v)|$.}

\paragraph{}{The problem was first studied by Ajtai et al. \cite{AjtaiANRSW98} as a special case of the general carpooling problem for hypergraphs posed by Fagin and Williams in \cite{FaginW83} and is motivated by fairness in scheduling. Indeed, the above problem exactly corresponds to a set of drivers who wish to carpool in pairs where each driver wishes to drive roughly half the time. If all the edges of the graph are known beforehand, a simple scheme of directing cycles yields a solution of discrepancy at most one. That said, the more natural setting is arguably when the edges arrive online. In this setting, edges arrive sequentially, and one must irrevocably decide on the orientation of an edge before observing the next edge in the sequence. In this model, \cite{AjtaiANRSW98} showed that in the presence of an adaptive adversary, any orientation algorithm must have $\Omega(n)$ discrepancy in the worst case.}

\paragraph{}{There are several approaches in the literature for circumventing such a lower bound. One, explored in \cite{AjtaiANRSW98,GuptaK0020}, considers a stochastic version of the problem, namely when edge arrivals are sampled from a known distribution. In such a setting, \cite{GuptaK0020} obtain a solution achieving $\polylog(n,m)$-discrepancy.
In a recent work, \cite{AlweissLS21, liu2021gaussian} improved this bound by showing that one can achieve discrepancy $O(\log(nm))$ in the presence of any oblivious adversary. This was further improved to $O(\sqrt{\log(mn)})$ discrepancy by Kulkarni et al. \cite{kulkarni2024optimal}, although their methods do not give an efficient algorithm for computing such a solution.

In fact, they proved this result for the more general $\OVB$ problem, where vectors $v_1, \dots, v_m \in \mathbb{R}^n$ with $\|v_i\|_2 = O(1)$
arrive online and must be assigned $\pm 1$ signs $\eps_1, \dots, \eps_m$ so as to minimize $\max_t \|\sum_{i \leq t} \eps_i v_i\|_{\infty}$. Note this is more general as the online graphical carpooling problem can be written as a special instance of the vector balancing problem in which the vectors are of the form $e_j - e_k$ for $j,k \in [n]$.}

\paragraph{}{An additional approach, mainly aimed at circumventing the adaptive lower bound, is by introducing recourse, i.e. allowing the player to change the orientation of a limited number of edges each time a new edge arrives.
We refer to the number of edges whose sign has changed as the \emph{recourse}.
With such a relaxation, one can consider an even harder model in which the player is required to maintain good discrepancy not only under edge arrivals (henceforth called \emph{insertion updates}) but also edge removals (henceforth called \emph{deletion updates}). 
This setting is usually dubbed the \emph{Fully-Dynamic} setting.
The work of \cite{gupta2022online} was the first to consider such a setting for both the graphical carpooling problem and the $\OVB$ problem.
For $\OVB$, they show that one can deterministically maintain a near-optimal discrepancy of $O(\sqrt{\log(n)})$ with amortized recourse of $O(n\log m)$ per update.
For the fully-dynamic, graphical carpooling problem, they design a deterministic algorithm achieving $O(\log^7n)$ discrepancy and amortized recourse of $O(\log^5 n)$ per update. 
While their algorithm for $\OVB$ is simple and elegant, their result in the carpooling setting involves clever and sophisticated machinery that maintains an expander decomposition of the graph and orients edges via a local search procedure.
In this work, we present both a simplification of their result for the graphical carpooling problem and an improvement upon their parameters. Namely, our main result is the following Theorem \ref{thm:main}.  }

\begin{theorem}
\label{thm:main}
    There exists a deterministic algorithm for the fully-dynamic, graphical carpooling problem that maintains a solution of discrepancy $3$ and worst case recourse $O(\log^2(n))$.
\end{theorem}

\paragraph{}{Our main technical tool and major difference to the work of \cite{gupta2022online} is noticing that maintaining a solution for graphs with no short cycles, i.e. a graph of high girth, is simple (cf. Lemma \ref{lem:min-out-deg}). For a general graph, we can then maintain a solution to a high girth subgraph $H$. If the addition of an edge creates a short cycle in $H$, we simple remove all edges in the cycle from $H$ and direct the cycle clockwise. 
For deletions, our algorithm for high girth graphs can handle the removal of any edge from $H$ (cf. Lemma \ref{lem:min-out-deg}). On the other hand, if we delete an edge $e$ from a short cycle $C$, then we remove every edge from the cycle $C$ and subsequently simulate inserting every edge in $C \setminus e$ to update our orientation.
Conceptually, we note that when the edges are viewed as vectors vis-{\`a}-vis the correspondence described in the reduction to the $\OVB$ problem short cycles correspond exactly to linear dependencies, which also form the backbone for the algorithm of \cite{gupta2022online} for the $\OVB$ problem that achieves $O(n \log(m))$ recourse.}

We note that while our algorithms are efficient, we do not pay close attention to their run times. That said, a closely related problem for which dynamic graph algorithms are well-studied is to minimize the maximum out-degree, rather than the discrepancy. Various algorithms achieve tradeoffs between update time,  maximum out-degree, and number of edge flips, typically in (pseudo) forests. 
Brodal and Fagerberg~\cite{BF99} showed that maintaining an orientation with bounded out-degree takes {\em amortized } $O({\log n})$ insertion time and {\em worst-case} $O(1)$ deletion time. 
Several works have given improved results and different tradeoffs, see, e.g. 
\cite{HTZ14,KKPS14,BB17,BenderKKP021} for more details.

\section{Preliminaries}\label{sec:prleim}
\textbf{Notation.} We will let $G = (V,E)$ denote an undirected graph with multi-edges. $H = (V,E)$ will denote a directed graph with multi-edges. An edge directed from $u$ to $v$ is denoted $uv$. For a $v \in V$, We let $\deg_G(v)$ denote the degree of $v$ in $G$, $\deg_H^-(v)$ denote the in-degree of $v$ in $H$, and $\deg_H^+(v)$ denote the out-degree of $v$ in $H$. Given a graph $G = (V,E)$ and a subset of edges $F \subseteq E$, we let $G[F]$ denote the graph induced by $F$. Given vertices $u,v \in V$, $u \rightsquigarrow v$ denotes the shortest path from $u$ to $v$. Finally, given a graph $G = (V,E)$ and an edge $e \in V \times V$ we let $G \setminus e$ and $G \cup e$ denote the graphs $(V, E \setminus \{e\})$ and $(V, E \cup \{e\})$, respectively.

With this, we now formally define the Fully-Dynamic, graphical carpooling problem.

\begin{definition}[Fully-Dynamic, Graphical Carpooling]\label{def:FDCP}
Starting from an empty graph $G_0=(V,E_0)$ with $E_0=\emptyset$, an adaptive adversary at each time $t$ chooses an update edge $e_t\in V\times V$ to add or remove, yielding the graph $G_t = G_{t-1} \cup e$ or $G_t = G_{t-1} \setminus e$ respectively. We are tasked with maintaining an orientation of $G_t$, $H_t$, such that $\max_t \max_{v \in V} |\deg_{H_t}^+(v) - \deg_{H_t}^-(v)|$ is minimized. The algorithm is allowed to change the orientation of edges after every update, and we refer to the total number edge orientation changes made by the algorithm as the recourse.
\end{definition}

We remark that one often assumes a bound on the number of updates performed by the adversary, but our algorithms and their guarantees will not depend on the number of updates performed by the adversary.

We also centrally use the girth of a graph.

\begin{definition}[Girth]\label{def:Girth}
    The girth of a graph $G = (V,E)$ is the length of shortest cycle in $G$.
\end{definition}

\section{Proof of Theorem \ref{thm:main}}\label{sec:proof_of_main}
\subsection{Handling High Girth Graphs}\label{ssec:high_girth_graphs}
In this section, we describe an algorithm for orienting edges in a graph that is always promised to have high girth. Formally, we'll show the following.

\begin{lemma}
\label{lem:sparse-case}
    Suppose edges arrive and depart online in such a way that the underlying graph $G = (V,E)$ always has girth strictly bigger than $2\log(n)$. Then there exists a deterministic algorithm that maintains a solution to the fully-dynamic graphical carpooling problem with discrepancy at most $3$ and worst case recourse $O(\log(n))$ per addition and $O(1)$ recourse per removal.
\end{lemma}

The main tool in the proof will be the following  lemma.

\begin{lemma}
\label{lem:min-out-deg}
    Suppose that edges arrive and depart online in such a way that the underlying graph $G = (V,E)$ has girth strictly bigger than $2 \log(n)$. Then there exists a deterministic algorithm that orients every edge such that $v \in V$ has out degree at most $2$ and has worse case recourse $O(\log(n))$ per addition and $O(1)$ recourse per removal.
\end{lemma}

We remark that \cite{gupta2014maintaining} proves this result for trees. Since graphs of high girth are ``locally tree-like,'' we can essentially use the same proof.

\begin{proof}
Let $G$ denote the underlying graph and $H$ the oriented graph. We now use the following simple procedure to achieve the guarantee: When an edge is deleted from $G$, do nothing. If an edge is added to $G$, orient it arbitrarily in $H$, say from $u$ to $v$. We then find the shortest path in $H$ from $u$ to a vertex $w$ with out-degree at most $1$ and flip every edge along the path, where we possibly have that $w$ is equal to $u$. It turns out that such a path will always exist, as we will argue below.

Clearly, this procedure has recourse at most $O(1)$ per removal. Additionally, any removal clearly does not increase the maximum out degree. On the other hand, it's simple to check that flipping edges along the path will also keep the maximum out-degree below $2$ so long as such a path exists. 

Thus, it remains to prove that a path to vertex of degree at most $1$ always exists and is short, as the length of the path will bound the recourse of an insertion. Towards a contradiction, suppose that there does not exist such a path of length at most $\log(n)$. We now show the following claim:

\begin{claim}
    $|\{w: \dist_H(u,w) = \ell\}| \geq 2^\ell$ for all $\ell \leq \log(n)$.
\end{claim}

\begin{proof}
We prove the statement by induction. The base case is trivial since by assumption $u$ has out-degree at least $2$. Now assume the statement holds for some $\ell \leq \log(n) -1$. Let $x,y \in \{w: \dist_H(u,w) \leq \ell\}$. Note that if $x$ and $y$ share a neighbor $z$ then $u \rightsquigarrow x \rightsquigarrow z \rightsquigarrow y \rightsquigarrow u$ contains a cycle of length $\leq 2 \log(n)$ in $G$, a contradiction. Thus, all vertices in $\{w: \dist_H(u,w) \leq \ell\}$ must have distinct out neighbors. Moreover, since the shortest path to a vertex with out-degree at most $1$ has length at least $\log(n) + 1$, every vertex in $|\{w: \dist_H(u,w) = \ell\}|$ must have out-degree at least $2$. Thus,
    \[|\{w: \dist_H(u,w) = \ell + 1\}| \geq 2|\{w: \dist_H(u,w) = \ell\}| \geq 2^{\ell + 1}\]
This completes the proof of the claim.     
\end{proof}

Applying the claim with $\ell = \log(n)$ yields $n$ vertices at distance $\log(n)$, but since $\dist_H(u,u) = 0$ there can be at most $n-1$ such vertices. So we have arrived at the desired contradiction.
\end{proof}

In order to avoid having multiple orientations floating around, we now make the the following definition:

\begin{definition}[Edge Labelling]\label{def:edge_labeling}
Given a graph $G = (V,E)$, we say an edge labelling $L: E \rightarrow V$ is a function that maps each edge $uv \in E$ to one its endpoints $u$ or $v$.
\end{definition}

Note that edge labelling are exactly equivalent to edge orientations. With this, we now turn to the proof of Lemma \ref{lem:sparse-case}. The high level idea will be to maintain an edge labelling $L$ over the graph such that every vertex $v \in V$ has at most two incident edges not labelled with $v$. We will then only be interested in balancing the edges incident to $v$ that are labelled according to $L$. Note that this task is relatively easy as each edge now only contributes to the imbalance of one vertex instead of two.

\begin{proof}[Proof of Lemma \ref{lem:sparse-case}]
    Let $G$ denote the underlying graph and $H$ denote the orientation that we will maintain. The key idea will be to maintain two invariants. First, we keep an edge labelling $L$ of $G$ such that every $v \in V$ has at most $2$ incident edges labelled according to some vertex in $V \setminus \{v\}$. Next, to orient edges, we enforce that for any vertex $v \in V$, we always have
        \begin{equation}
        \tag{$\star$}
        \label{invariant}
            \bigg | \big|\{vu :L(vu) = v \land vu \in E(H)\} \big| - \big |\{uv:L(uv) = v \land uv \in E(H)\} \big|\bigg | \leq 1.
        \end{equation} 
    Intuitively, we are balancing the edges incident to $v$ that $L$ labels with $v$. As there are at most two other edges incident to $v$ and labelled with a different vertex, the imbalance of edges with label $v$ will differ by at most $2$ from the discrepancy at $v$. Thus, maintaining both invariants would lead to an orientation of discrepancy of at most $3$, as desired.

    It remains to show that these invariants can be maintained. When an edge arrives or departs, we update the labelling $L$ according to Lemma \ref{lem:min-out-deg}. If an edge $uv$'s label changed (or that edge departed or arrived), then we change the minimum number of edges to fix Equation (\ref{invariant}) at $u$ and $v$. More formally, note that we can always fix the invariant as each edge contributes to Equation (\ref{invariant}) for a unique vertex in $V$. If after updating the labelling to $L'$ we have that Equation (\ref{invariant}) at vertex $u$ is equal to $k$, then we flip the orientation of $\lfloor \frac{k}{2} \rfloor$ edges in such a way to decrease the imbalance to have value at most $1$.

    To bound the recourse, consider Equation (\ref{invariant}) under the updated labelling $L'$ at $v$ before changing the orientation of any of the edges, we then have that 
        \begin{align*}
        \bigg | \big|\{vu &:L'(vu) = v \land vu \in E(H)\} \big| - \big |\{uv:L'(uv) = v \land uv \in E(H)\} \big|\bigg | \\
        &\quad \quad \quad - \bigg | \big|\{vu :L(vu) = v \land vu \in E(H)\} \big| - \big |\{uv:L(uv) = v \land uv \in E(H)\} \big|\bigg | \\
         &\leq \bigg | \big|\{vu :L'(vu) = v \land vu \in E(H)\} \big| - \big|\{vu :L(vu) = v \land vu \in E(H)\} \big| \\
         & \quad \quad \quad + \big|\{uv :L(uv) = v \land uv \in E(H)\} \big| - \big |\{uv:L'(uv) = v \land uv \in E(H)\} \big|\bigg |  \\
         &\leq \bigg | \big|\{vu :L'(vu) = v \land vu \in E(H)\} \big| - \big|\{vu :L(vu) = v \land vu \in E(H)\} \big| \bigg| \\
         & \quad \quad \quad + \bigg| \big|\{uv :L(uv) = v \land uv \in E(H)\} \big| - \big |\{uv:L'(uv) = v \land uv \in E(H)\} \big|\bigg | \\
         &\leq |\{e \in E(H) : L'(e) \not = v \land L(e) = v \}| + |\{e \in E(H) : L'(e) = v \land L(e) \not = v \}|
        \end{align*}
    Thus, to maintain invariant Equation (\ref{invariant}) at $v$, we need to make at most 
        \begin{align*}
        \left \lfloor \frac{ |\{e \in E(H) : L'(e) \not = v \land L(e) = v \}| + |\{e \in E(H) : L'(e) = v \land L(e) \not = v \}| + 1}{2} \right \rfloor \\
        \leq |\{e \in E(H) : L'(e) \not = v \land L(e) = v \}| + |\{e \in E(H) : L'(e) = v \land L(e) \not = v \}|
        \end{align*}
    changes. Summing over all vertices $v \in V$ then yields that the total number of changes is bounded by at most twice the number of changes made to the labeling $L$. Using Lemma \ref{lem:min-out-deg}, we then get that we make $O(\log(n))$ changes per addition and $O(1)$ changes per removal.
\end{proof}

\subsection{Carpooling for General Graphs}
We now leverage our algorithm for high girth graphs to get a solution for the case of arbitrary graphs. The proof will closely follow the sketch presented in the introduction.

\begin{algorithm}[H]
\addtolength\linewidth{-2em}

\vspace{0.5em}

\textbf{Input:} A graph $G = (V,E)$, an orientation $H$ of $G$, a set $E_{girth} \subseteq E(G)$ such that $G[E_{girth}]$ has girth at least $2 \log(n) + 1$, a set $\calC \subseteq 2^E$ such that each element of $\calC$ is a cycle of length at most $2 \log(n)$ and $E_{girth} \sqcup \calC$ partitions $E(G)$, and an edge $e \in V \times V$ to add

\textbf{Output:} An updated graph $G$, orientation $H$, and partition $E_{girth} \sqcup \calC$
\\[0.25em]

\textsc{Add Edge}($G, H, E_{girth}, \calC,e$):
\begin{enumerate}
    \item If $G[E_{girth} \cup \{e\}]$ has girth at least $2 \log(n) + 1$:
        \begin{enumerate}
            \item Update $H$ by adding $e$ to $H[E_{girth}]$ via Lemma \ref{lem:sparse-case}
            \item Return $G \cup e, H, E_{girth} \cup e, \calC$
        \end{enumerate}
    \item Let $C$ denote any cycle of length at most $2 \log(n)$ containing $e$ in $E_{girth} \cup \{e\}$
    \item For each edge $f \in C \setminus e$, remove $f$ from $H[E_{girth}]$ via Lemma \ref{lem:sparse-case}.
    \item Add $C$ to $H$ and orient the edges in $C$ clockwise.
    \item Return $G \cup e, H, E_{girth} \setminus C, \calC \cup C$
\end{enumerate}
\caption{Update Edge Arrivals}
\label{alg:arrivals}
\end{algorithm}

\begin{algorithm}[H]
\addtolength\linewidth{-2em}

\vspace{0.5em}

\textbf{Input:} A graph $G = (V,E)$, an orientation $H$ of $G$, a set $E_{girth} \subseteq E(G)$ such that $G[E_{girth}]$ has girth at least $2 \log(n) + 1$, a set $\calC \subseteq 2^E$ such that each element of $\calC$ is an cycle of length at most $2 \log(n)$ and $E_{girth} \sqcup \calC$ partitions $E(G)$, and an edge $e \in E$ to remove

\textbf{Output:} An updated graph $G$, orientation $H$, and partition $E_{girth} \sqcup \calC$
\\[0.25em]

\textsc{Remove Edge}($G, H, E_{girth}, \calC,e$):
\begin{enumerate}
    \item If $e \in E_{girth}$:
        \begin{enumerate}
            \item Update $H$ by removing $e$ from $H[E_{girth}]$ via Lemma \ref{lem:sparse-case}
            \item Return $G \setminus e, H, E_{girth} \setminus e, \calC$
        \end{enumerate}
    \item Let $C \in \calC$ denote the cycle containing $e$ 
    \item Remove $C$ from $\calC$, $G$ and $H$
    \item For each $f \in C \setminus \{e\}$: 
    \begin{enumerate}
        \item Update $(G, H, E_{girth}, \calC)$ by running \textsc{Add Edge}$(G, H, E_{girth}, \calC, f)$
    \end{enumerate}
    \item Return $G, H, E_{girth}, \calC$
\end{enumerate}
\caption{Update for Edge Departures}
\label{alg:departures}
\end{algorithm}

\begin{proof}[Proof of Theorem \ref{thm:main}]
    As before we let $G$ be the underlying graph and $H$ the oriented graph. Our algorithm will will partition the edges into a high girth graph and a set of $O(\log(n))$ length cycles. Initially, there are no cycles and the high girth graph is empty. To add edges, we then use Algorithm \ref{alg:arrivals} to update our state. To remove edges, we use Algorithm \ref{alg:departures}.
    
    We now bound the recourse of our procedures. Note that if we add an edge and it doesn't create a short cycle then there is $O(\log(n))$ recourse by Lemma \ref{lem:sparse-case}. If we add an edge and it creates a short cycle, then we remove $O(\log(n))$ edges from the high girth graph which again results in $O(\log(n))$ recourse by Lemma \ref{lem:sparse-case}.

    If we remove an edge from the high girth graph, then this has recourse $O(1)$ by Lemma \ref{lem:sparse-case}. On the other hand, removing an edge from a cycle causes $O(\log(n))$ additions, each of which may have recourse $O(\log(n))$. So there is $O(\log^2(n))$ recourse in this case. 
    
    Finally, to analyze the discrepancy, note that the edges in $\mathcal{C}$ contribute $0$ to the discrepancy of any vertex. On the other hand, by Lemma \ref{lem:sparse-case}, we have that the edges in $E_{girth}$ have discrepancy at most $3$. Thus, the total discrepancy is at most $3$, as claimed.
\end{proof}

\bibliographystyle{alpha}
\bibliography{references}

\end{document}